\documentclass[letterpaper]{conm-p-l}

\copyrightinfo{2017}{Per von Soosten and Simone Warzel}

\usepackage{mathtools}
\usepackage[hidelinks, citecolor=black]{hyperref}
\hypersetup{colorlinks=true, urlcolor = black, linkcolor=black}

\newtheorem{theorem}{Theorem}[section]
\newtheorem{corollary}[theorem]{Corollary}
\numberwithin{equation}{section}
\newcommand{\rr}{{\mathbb{R}}}
\newcommand{\nn}{{\mathbb{N}_0}}
\newcommand{\zz}{{\mathbb{Z}}}
\newcommand{\ee}{{\mathbb{E}\,}}
\newcommand{\pp}{{\mathbb{P}}}
\newcommand{\oh}{{\mathcal{O}}}
\newcommand{\supp}{{\operatorname{supp }}}
\newcommand{\ketbra}[1]{{{|#1\rangle\langle #1|}}}
\newcommand{\beq}[1]{\begin{equation} \label{#1}}
\newcommand{\eeq}{\end{equation}}
\newcommand{\p}{\mathcal{P}}

\begin{document}
\title[Hierarchical operators]{Singular spectrum and recent results on hierarchical operators}
\author{Per von Soosten}
\address{Zentrum Mathematik, TU M\"{u}nchen\\
Boltzmannstr. 3, 85747 Garching, Germany}
\email{vonsoost@ma.tum.de}
\author{Simone Warzel}
\address{Zentrum Mathematik, TU M\"{u}nchen\\
Boltzmannstr. 3, 85747 Garching, Germany}
\email{warzel@ma.tum.de}
\thanks{This work was supported by the DFG (WA 1699/2-1).}
\begin{abstract} We use trace class scattering theory to exclude the possibility of absolutely continuous spectrum in a large class of self-adjoint operators with an underlying hierarchical structure and provide applications to certain random hierarchical operators and matrices.  We proceed to contrast the localizing effect of the hierarchical structure in the deterministic setting with previous results and conjectures in the random setting. Furthermore, we survey stronger localization statements truly exploiting the disorder for the hierarchical Anderson model and report recent results concerning the spectral statistics of the ultrametric random matrix ensemble.
\end{abstract}
\maketitle

\section{Hierarchical operators} \label{abstractsection}
Random Hamiltonians with hierarchical structure provide highly simplified and analytically tractable toy models for the Anderson localization-delocalization transition, whose full version is beyond the reach of the currently available mathematical machinery. In this article, we define a notion of abstract hierarchical operator and use trace class scattering theory to show that the hierarchical structure forbids the formation of absolutely continuous spectrum even in the deterministic setting. For deterministic operators, such a result is optimal regarding the spectral type since changing even a single potential value can produce singular continuous spectrum~\cite{MR1298942}. The secondary aim of this article is to survey recent results concerning localization in two concrete disordered hierarchical models, the hierarchical Anderson model introduced by Bovier in~\cite{MR1063180} and the ultrametric random matrix ensemble of Fyodorov, Ossipov and Rodriguez~\cite{1742-5468-2009-12-L12001}. This is done in Sections \ref{hamsection} and \ref{uesection}, where we also show how to capture these models in the abstract framework and apply our main result, thereby proving the absence of absolutely continuous spectrum for the infinite-volume Hamiltonians of these models.

The idea of exploiting the stability of the absolutely continuous spectrum to rule out its presence altogether goes back at least to the work of Dombrowski~\cite{MR757504, MR808928, MR891145} and was rediscovered and popularized in Simon and Spencer's treatment of one-dimensional Schr\"{o}dinger operators with potential barriers~\cite{MR1017742}. These arguments are based on comparing the operator in question with a direct sum of finite-dimensional matrices, which has only pure-point spectrum. If the resolvents agree up to a trace class perturbation, the Kato-Rosenblum theorem~\cite{MR529429} asserts that the absolutely continuous spectra of the two operators are equal, which is to say, both are empty. We follow the same approach here, but the hierarchical structure allows the direct comparison of the two operators without passing to the resolvents first.

Let us now turn to a detailed definition of hierarchical operators and the statement and proof of our main result. Inspired by~\cite{MR2276652}, we define a hierarchy in some (necessarily countable) set $X$ to be a sequence $\{\p_r\}_{r=0}^\infty$  of partitions of $X$ such that:
\begin{itemize}
\item the members of $\p_0$ are finite sets,
\item each member of $\p_r$ is a finite union of members of $\p_{r-1}$, and
\item for any $x, y \in X$ there exists $r \geq 0$ such that $x$ and $y$ lie in a common member of $\p_r$.
\end{itemize}
Thus, if $x \in X$ and $r \geq 0$, there exists a unique member of $\p_r$ containing $x$, which we will denote by $B_r(x)$. We say that an operator $H$ on $\ell^2(X)$ is hierarchical if it is of the form
\beq{hdef} H = \sum_{r \geq 0} \sum_{B \in \p_r}  H(B)
\eeq
where $H(B)$ is a self-adjoint operator on $\ell^2(B)$ for every $B \in \p_r$ and
\beq{trassumption}\sum_{r \geq 1} \|H(B_r(x))\|_1 < \infty\eeq
for all $x \in X$. In~\eqref{trassumption},  $\|\cdot\|_p$ is the Schatten $p$-norm, i.e., the $\ell^p$-norm of the singular values. The requirement~\eqref{trassumption} is not strictly necessary for~\eqref{hdef} to make sense as an operator on $\ell^2(X)$, but we will need it to classify the spectrum, and there is no real loss of generality in assuming it from the beginning. The action of $H$ on the function
\[\delta_x(y) = \begin{cases} 1 & \mbox{ if } y=x\\0& \mbox{ otherwise}\end{cases}\]
depends only on those summands of~\eqref{hdef} with $x \in B$, that is, $H\delta_x = S_x \delta_x$ with
\[S_x = \sum_{r \geq 1} H(B_r(x)).\]
By~\eqref{trassumption}, the sum defining $S_x \delta_x$ converges for all $x \in X$ so the domain of $H$ contains at least the functions with finite support. In the main new result of this article we prove that $H$ defines a self-adjoint operator without any absolutely continuous spectrum if the couplings $H(B_r(x))$ become sufficiently weak as $r \to \infty$, by showing that the removal of $S_x$ completely disconnects the system.

\begin{theorem} \label{mainthm} Every hierarchical operator $H$ is self-adjoint with $\sigma_{ac}(H) = \emptyset$.
\end{theorem}
\begin{proof}
Setting $X_0 = B_0(x)$, there exist $\{x_r\}$ in $X$ such that $X_r =  B_{r-1}(x_r)$ is a disjoint sequence of finite sets satisfying
\[B_r(x) = B_{r-1}(x) \cup X_r, \quad B_{r-1}(x) \cap X_r = \emptyset\]
for all $r \geq 1$. This means that
\[X = \bigcup_{r \geq 0} B_r(x) = \bigcup_{r \geq 0} X_r\]
and hence
\[\ell^2(X) = \bigoplus_{r \geq 0} \ell^2(X_r) \]
is a direct sum of finite-dimensional subspaces. If $y \in X_r$, then
\[H \delta_y = \sum_{r \geq 0} H(B_r(y)) \delta_y\]
and therefore
\[F \delta_y \vcentcolon =  (H - S_x)\delta_y = \sum_{s=0}^{r-1}H(B_s(y)) \delta_y\]
since $X_r \subset B_r(x)$. Thus $\ell^2(X_r)$ is an invariant subspace for $F$ because it is an invariant subspace of $H(B_s(y))$ for every $s \le r-1$, and this proves that $F$ is a direct sum of finite-dimensional matrices. Since $S_x$ is trace class, $H = F + S_x$ is self-adjoint on the domain of $F$ and
\[\sigma_{ac}(H)  = \sigma_{ac}(F) = \emptyset\]
by the Kato-Rosenblum theorem.
\end{proof}

\section{Hierarchical Schr\"odinger operators} \label{hamsection}
We now turn to the definition introduced in~\cite{MR1063180} of a hierarchical analogue of the finite-difference Laplacian on $\zz^d$. To capture the effect of nearest-neighbor hopping, or local averaging, it is natural to choose
\[H(B) = p_r \ketbra{\varphi_B},\]
if $r \geq 1$ and $B \in \p_r$, where
\[\varphi_B = |B|^{-1/2} 1_B\]
is the normalized maximally extended state in $\ell^2(B)$. Moreover, the decay in the interaction strength over long distances may be encoded in the requirement that the coefficients $p_r$ satisfy
\[\sum_{r=1}^\infty |p_r| < \infty.\]
We leave the choice of $H(B)$ with $B \in \p_0$ open for now by setting $H(B) = 0$ in this case, and define the hierarchical Laplacian as $\Delta$ as the hierarchical operator $H$ furnished by~\eqref{hdef} with this specific choice of $H(B)$.

An important quantity associated with a hierarchical Laplacian is the spectral dimension, whose definition is motivated by matching the decay of the density of states near the band edge to the lattice case. The precise definition,
\[d_s  =  \lim_{\lambda \downarrow 0 } \ \frac{\ln \ \langle \delta_k , 1_{[\lambda_\infty-\lambda, \lambda_\infty]}(\Delta) \delta_k \rangle }{\ln \sqrt{\lambda}}\]
with $\lambda_\infty = \sup \sigma(\Delta)$, is equal to the spatial dimension $d$ if $\Delta$ is replaced by the $d$-dimensional lattice Laplacian. The spectral dimension can be computed explicitly for a large class of hierarchies $\{\p_r\}$ and coefficients $p_r$ (see~\cite{MR2276652}).

A hierarchical Schr\"{o}dinger operator is the sum of a hierarchical Laplacian and a potential,
\[H = \Delta + V,\]
with
\[V = \sum_{x \in X} V_x \ketbra{\delta_x}\]
for $V_x \in \rr$. We note that this $H$ is still hierarchical in the sense of~\eqref{hdef} if we choose $\p_0$ to be the trivial partition, whose members are singletons, and set
\[H(\{x\}) = V_x.\]
Since $\{p_r\}$ was assumed to be a summable sequence, it is clear that $\|S_x\|_1 < \infty$ for any $x \in X$, and thus Theorem \ref{mainthm} answers the most immediate question concerning the localization properties of $H$.
\begin{corollary}\label{hamacspectrum} Every hierarchical  Schr\"odinger operator satisfies $\sigma_{ac}(H) = \emptyset$.
\end{corollary}

To obtain the stronger localization results that follow, the presence of true disorder is crucial, and we shall henceforth require that the $\{V_x\}$ are drawn independently from a common density $\rho \in L^\infty$, thereby obtaining the usual hierarchical Anderson model. Furthermore, it is often necessary to have some fixed decay rate of $p_r$ in terms of the partitions $\p_r$, so we focus on an explicit model, whose configuration space is $X = \nn$, and whose hierarchy is given in terms of the partitions $\{\p_r\}$ defined by
\beq{concreteP}\nn = \{0, ..., 2^{r}-1\} \cup \{2^r, ..., 2\cdot 2^r - 1\} \cup ...\eeq
It is now easiest to take $p_r =  \epsilon \, 2^{-cr}$ for some $\epsilon, c \in (0,\infty)$, as we shall do here, and one may check that the spectral dimension is then
\[d_s  = \frac{2}{c}.\]
In particular, the assumption $c > 0$ amounts to the requirement that the spectral dimension be finite.

A more complete characterization of $\sigma(H)$ was first given by Molchanov~\cite{MR1463464} and Kritchevski~\cite{MR2352276, MR2276652}, who proved that $H$ almost surely has only pure point spectrum, provided that $d_s < 4$ is small enough (i.e. $c > 1/2$ is large enough) or that $\rho$ is a mixture of Cauchy distributions with barycenters strictly separated from the real axis. Their argument, which is based on approximation by the truncations
\[H_s = \sum_{r = 0}^s \sum_{B \in \p_r} H(B)\]
and the Simon-Wolff criterion~\cite{MR820340}, was extended to the full parameter range $c > 0$ and $\rho \in L^\infty$ in~\cite{vonSoosten2017}. Moreover, it is possible to show that the spatial decay of the eigenfunctions of $H$ with respect to the metric
\beq{hmetric} d(j,k) = \min  \left\{ r \geq 0 \, | \, \mbox{$j $ and $ k $ belong to a common member of $\mathcal{P}_r$} \right\}\eeq
is essentially controlled by the decay of $p_r$.

\begin{theorem}[\cite{vonSoosten2017}]\label{hamppspectrum}
The  spectrum of $H $ is almost surely of pure-point type with eigenfunctions satisfying
\beq{hamefbound}\sum_{k \in \nn} 2^{\frac{c}{4} d(0,k)} |\psi_E(k)|^2 < \infty
\eeq
for any $ E \in \sigma(H) $. 
\end{theorem}

Via the RAGE theorem~\cite{MR529429}, Theorem \ref{hamppspectrum} has the dynamical implication that, for almost every realization of the disorder, the quantum probability of a particle ever leaving a finite radius $R$ vanishes asymptotically as $R \to \infty$. However, the complete lack of control of the implied amplitude of the wave function in~\eqref{hamefbound} makes this an assertion of a fundamentally qualitative nature and does not completely rule out the appearance of more subtle delocalization effects in finite volumes. Indeed, prior to the work~\cite{vonSoosten2017}, but after the appearance of~\cite{MR2352276, MR2276652, MR1463464}, Metz et.\ al.\ ~\cite{PhysRevB.88.045103} analyzed the hierarchical Anderson model numerically in large finite volumes by relying on renormalization group equations (see also~\cite{monthusgarel}) whose effect on the parameters is
\[\mathcal{R} (\{p_r\}_{r \geq 1}, \rho) = \left( (p_{r+1})_{r\geq 1} , T_{p_1} \rho \right),\]
where $T_p \rho $ is the probability density of
\[\left(\frac{1}{2V} + \frac{1}{2 V^\prime} \right)^{-1} + p,\]
and $V$ and $V^\prime$ are drawn independently from $\rho$. If $c < 1$ and $\rho$ is the density of a sufficiently weak Gaussian random variable, the authors of~\cite{PhysRevB.88.045103} claimed delocalization at the special energy $E = \sum p_r$ in the sense of the inverse participation ratios satisfying
\beq{metzclaim}\lim_{n \to \infty} \frac{1}{|B_n| \nu(E)} \ee \left[\sum_{\lambda \in \sigma(H_n)} \|\psi_\lambda\|_4 \tilde{\delta}_{E - \lambda} \right] = 0
\eeq
($\tilde{\delta}_{E-\lambda}$ being a suitably regularized Dirac delta). In~\eqref{metzclaim}, $\nu$ denotes the infinite-volume density of states defined by
\[\ee \langle \delta_0, f(H) \delta_0 \rangle = \int \! f(E) \nu(E) \, dE\]
and $\{\psi_\lambda\}$ are the normalized eigenfunctions of the finite-volume Hamiltonian
\[H_n = 1_{B_n(0)} H 1_{B_n(0)}.\]

Nevertheless, the renormalization dynamics $\mathcal{R}$ appear to effectively drive the Hamiltonian into a regime of high disorder, where strong finite-volume localization bounds apply generically. In fact, for any interval $I \subset \rr$ one can prove that there exists $\delta > 0$ such that
\beq{assumption} \sup_{E \in I }  \|T_{ p_r} ... T_{  p_1} \varrho(\cdot + E)\|_\infty = \mathcal{O} (2^{(c- \delta) r})
\eeq
whenever:
\begin{itemize}
\item $d_s < 2$,
\item $V$ is Gaussian and $d_s < 4$, or
\item $V$ has a Cauchy component and $d_s < \infty$.
\end{itemize}
This observation shows that the strength of the random potential does not change significantly in comparison to the decrease in the strength of the kinetic term,
\[(\mathcal{R}p)_r = p_{r+1} = 2^{-c} p_r,\]
and can be used to prove localization in terms of the eigenfunction correlator
\[Q_n(j,k; I) = \sup \, |\langle \delta_k, f(H_n) \delta_j \rangle | ,\]
where the supremum ranges over those $f \in C_0$ with $\supp f \subset I$ and $\|f\|_\infty \le 1$. 
\begin{theorem}[\cite{vonSoosten2017}] \label{hamecthm} If~\eqref{assumption} is satisfied in a bounded interval $I \subset \rr$, then there exist $C ,\mu \in (0,\infty) $ such that
\[\sup_{n \in \nn } \sup_{j \in \nn}   \sum_{k \in \nn} 2^{\mu d(j,k) } \ee[ Q_n(j,k; I)] \leq C | I | .\]
\end{theorem}
It follows that
\[\sum_{k : d(j,k) \geq R} \ee |\langle \delta_k, 1_I(H) e^{itH} \delta_j \rangle | ^2 \leq C \, 2^{-\mu R},\]
which gives a quantitative averaged bound on the quantum probability that a particle, which was started at $j \in \nn$ and subsequently filtered by energy, ever leaves $B_R(j)$. Theorem \ref{hamecthm} is not compatible with the claims of~\cite{PhysRevB.88.045103} since it implies lower bounds for the inverse participation ratios of the eigenfunctions; the details can be found in~\cite{vonSoosten2017}.

A further canonical method of probing localization properties in finite systems consists of examining the statistics of the energy levels, i.e.\ the local behavior of the eigenvalues on the microscopic scale. This is best quantified by the rescaled eigenvalue point process
\beq{pointprocessdef} \mu_n(f) = \sum_{\lambda \in \sigma(H_n)} f(2^n(\lambda - E)).\eeq
In terms of this random measure, the renormalization group approach lets one prove the following theorem, which was previously proved by Kritchevski~\cite{MR2413200} for $d_s < 1$.
\begin{theorem}[\cite{vonSoosten2017}]\label{hampoissonthm} Suppose~\eqref{assumption} is satisfied in an open set $I \subset \rr$ and $E \in I$ is a Lebesgue point of $\nu$. Then, $\mu_n$ converges in distribution to a Poisson point process with intensity $\nu(E)$ as $n \to \infty$.
\end{theorem}

It remains an open problem to prove~\eqref{assumption} for more general densities $\rho$, and we hope to inspire further work in this direction. Still, even without a general theorem of this kind, our results leave no doubt that the hierarchical approximation in finite dimension is too crude to capture the Anderson transition on the lattice.

\section{Ultrametric ensemble} \label{uesection}
In this section we continue to consider the hierarchy~\eqref{concreteP} and retain the notation $d$ for the restriction of the hierarchical metric~\eqref{hmetric} to the finite volume $B_n = \{1,2,..., 2^n\}$. The second model we wish to study in this article, the ultrametric ensemble, consists of hierarchical matrices on $\ell^2(B_n)$ whose off-diagonal entries are also random variables with variances satisfying power-law decay in $d(k, \ell)$. We are thus led to consider
\beq{uedef} H_n = \sum_{r =0}^n 2^{-\frac{(1+c)}{2}r} \Phi_{n,r},\eeq
where $c \in \rr$ is a real parameter and the entries of $\Phi_{n,r}$ are independent centered real Gaussian random variables with variances
\[\ee \left| \langle \delta_\ell, \Phi_{n,r} \delta_k \rangle \right|^2 = 2^{-r} \begin{cases} 2 & \mbox{ if } d(k,\ell) = 0\\ 1 & \mbox{ if } 1 \le d(k,\ell) \le r\\ 0 & \mbox{ otherwise. }\end{cases}\]
With this convention, each $\Phi_{n,r}$ is a direct sum of $2^{n-r}$ random matrices drawn independently from the Gaussian Orthogonal Ensemble (GOE), rather than the Gaussian Unitary Ensemble used in~\cite{1742-5468-2009-12-L12001}, but our results are not based on any additional symmetries and are valid in both cases.

The definition~\eqref{uedef} shows that $H_n$ has independent entries whose typical magnitudes are approximately
\[ \left( \ee \left| \langle \delta_\ell, H_n \delta_k \rangle \right|^2\right)^{1/2} \approx 2^{-\frac{2+c}{2}d(\ell,k)},\]
and hence the ultrametric ensemble interpolates between a perfectly localized random potential and a completely delocalized Wigner matrix as the parameter $c$ varies in $\rr$. Indeed, this ensemble is a hierarchical analogue of the Power-Law Random Band Matrices~\cite{PhysRevE.54.3221} in the sense which was first introduced to statistical mechanical models by Dyson~\cite{MR0436850}. The authors of~\cite{1742-5468-2009-12-L12001} argue for a localization-delocalization transition in the eigenfunctions of $H_n$ as the parameter varies from $c > 0$ to $c < 0$ with a theoretical physics level of rigor and support this claim with numerical simulations.

An infinite-volume version of the ultrametric ensemble with parameter $c \in \rr$ may be defined formally using the prescription~\eqref{hdef}, by letting each $H(B)$ be an independent and appropriately scaled GOE, that is, 
\[\ee \left|\langle \delta_x, H(B) \delta_y \rangle\right|^2 = |B|^{-(1 + c)} \frac{1 + \delta_{xy}}{|B|}.\]
Of course, this $H$ can only be defined as an operator on $\ell^2(X)$ if $c > -1$, and in this case one may easily check that $H\delta_x = S_x\delta_x$ converges almost surely to an element of $\ell^2(X)$. As with the hierarchical Anderson model, we may apply Theorem \ref{mainthm} to exclude absolutely continuous spectrum when $S_x$ is trace class.

\begin{corollary} If $c > 1$, the infinite-volume ultrametric operator $H$ satisfies $\sigma_{ac}(H) = \emptyset$ almost surely.
\end{corollary}
\begin{proof} By applying Jensen's inequality twice we get the inequality
\begin{align*}\ee \|H(B_r(x))\|_1 &\le |B_r(x)|^{1/2} \ee \|H(B_r(x))\|_2\\
&\le |B_r(x)|^{1/2} \left(\sum_{y, y^\prime \in B_r(x)} \ee  \left|\langle \delta_{y^\prime}, H(B_r(x)) \delta_y \rangle\right|^2 \right)^{1/2}\\
&= |B_r(x)|^{\frac{1-c}{2}}
\end{align*}
for the expected trace norm of $H(B_r(x))$. Because each member of $\p_{r+1}$ is the union of two members of $\p_r$, we have $|B_r(x)| \geq 2^r$ so Markov's inequality shows that
\[\sum_{r \geq 0} \pp\left( \|H(B_r(x))\|_1 \geq  |B_r(x)|^{\frac{1-c}{4}} \right) <\infty,\]
provided $c > 1$. Therefore, the Borel-Cantelli lemma implies that there almost surely exists some constant $C<\infty$ such that
\[\|H(B_r(x))\|_1 \le C |B_r(x)|^{\frac{1-c}{4}} \le C \, 2^{\frac{1-c}{4}r}\]
for all $r \geq 0$. Hence
\[\|S_x\|_1 \le \sum_{r \geq 0} \|H(B_r(x))\|_1 < \infty\]
and Theorem \ref{mainthm} implies the result.
\end{proof}

We remark that Theorem \ref{mainthm} fails to cover the entirety of the regime in which the infinite volume operator is well-defined, and it is an open question whether the formation of continuous spectrum is possible in the remaining parameter range.

To obtain more detailed results in finite volumes, we now focus on the eigenvalue statistics of $H_n$. The effect of a GOE perturbation $\Phi$ on the spectrum of an arbitrary symmetric $N \times N$ matrix $H_0$ can be described dynamically by thinking of $H_0 + \Phi$ as $H_0 + \Phi(t)$ where $\Phi(t)$ is the symmetric stochastic matrix flow 
\[\langle \delta_\ell, \Phi(t) \delta_k \rangle = \sqrt{\frac{1 + \delta_{k\ell}}{N}} B_{k \ell}(t)\]
and $\{B_{k\ell}\}$ is a symmetric array of independent standard Brownian motions. The observation of Dyson~\cite{MR0148397} was that the eigenvalues $\lambda_1(t) \le ... \le \lambda_N(t)$ then undergo Dyson Brownian motion, that is, evolve according to the stochastic differential equation
\[d\lambda_j(t) = \sqrt{\frac{2}{N}}dB_j(t) + \frac{1}{N}\sum_{i \neq j} \frac{dt}{\lambda_j(t) - \lambda_i(t)}.\]
The relevance of this discussion to the ultrametric ensemble can be seen from the relation
\[H_n = \sum_{r =0}^{n-1} 2^{-\frac{(1+c)}{2}r} \Phi_{n,r} + 2^{-\frac{(1+c)}{2}r}\Phi_{n,n} = H_{n-1} \oplus H_{n-1}^\prime + \Phi(t),\]
where the system size $N$ is now $2^n$, $t = 2^{-(1+c)n}$, and $H_{n-1}^\prime$ is an independent copy of $H_{n-1}$. One may thus construct the spectrum of $H_n$ by initializing $\sigma(H_0)$ to consist of a single standard normal random variable and following the recursion
\begin{enumerate}
\item Sample an independent copy $\sigma(H_{k-1}^\prime)$ of $\sigma(H_{k-1})$
\item Let $\sigma(H_{k})$ be the evolution of $\sigma(H_{k-1}) \cup \sigma(H_{k-1}^\prime)$ under Dyson Brownian motion with duration $t = 2^{-(1+c)k}$.
\end{enumerate}

In his original paper, Dyson motivated the conjecture that the Dyson Brownian motion of $N$ particles locally equilibrates at times greater than $t = N^{-1}$, that is, the particle configuration $\{\lambda_j(t)\}$ becomes extremely rigid locally and the $k$-point correlation functions are very close to those of the GOE. Thus, one does not expect the addition of independent randomness in step (1) to significantly affect the spectrum of $H_{k}$ on small enough scales when $c < 0$ and the level statistics should asymptotically agree with the GOE. On the other hand, if $c > 0$, the Dyson Brownian motion in step (2) is not running for a long enough time to compensate the additional fluctuations of step (1) and the resulting level statistics are only a slight perturbation of a point process with many independent components. This second point is the basic idea behind the proof of the following theorem of~\cite{resflow}. As before, $\nu$ denotes the density of states in the infinite volume and $\mu_n$ is the random measure~\eqref{pointprocessdef}.
\begin{theorem}[\cite{resflow}]\label{uepoissonthm} Suppose $c > 0$ and $E \in \rr$ is a Lebesgue point of $\nu$. Then, $\mu_n$ converges in distribution to a Poisson point process with intensity $\nu(E)$ as $n \to \infty$.
\end{theorem}

Similar considerations for the off-diagonal spectral measures also allow one to prove localization estimates for the eigenfunctions of $H_n$.
\begin{theorem}[Eigenfunction localization]\label{localizationthm} Suppose $c > 0$ and let $E \in \rr$. Then, there exist $w, \mu, \kappa > 0$, $C<\infty$, and a sequence $m_n$ with $n - m_n \to \infty$ such that for every $x \in B_n$ the $ \ell^2$-normalized eigenfunctions satisfy
\[\pp\left( \sum_{y  \in B_n \setminus B_{m_n}(x)} Q_n(x,y;W) > 2^{-\mu n} \right) \le C \,2^{-\kappa n}\]
with
\[W = \left[E_0 - 2^{-(1-w)n}, E_0 + 2^{-(1-w)n}\right].\]
\end{theorem}

For times $t \gg N^{-1}$, Landon, Sosoe and Yau~\cite{landonsosoeyauarxiv} have recently proved the fixed energy universality (and more) of the Dyson Brownian motion for very general regularly spaced initial conditions. A key feature of this result is that as $t$ approaches $\oh(1)$, the scale of the required local law may also be relaxed towards $\oh(1)$. If $c < -1$, then the $\ell^2$-norm of $H_n\delta_0$,
\[Z_{n,c} \vcentcolon= \left(\sum_{\ell \in B_n} \ee \left| \langle \delta_\ell, \left(\sum_{r =0}^n 2^{-\frac{(1+c)}{2}r} \Phi_{n,r}\right) \delta_0\rangle \right|^2\right)^{1/2},\]
is approximately $2^{-\frac{1+c}{2}n}$. Therefore, the rescaled Hamiltonian $Z_{n,c}^{-1} \, H_n$, whose spectrum lies in $[-2,2]$ with high probability, has spread
\[M \vcentcolon= \left(\max_{k,\ell \in B_n} \, \ee |\langle \delta_\ell, H_n\delta_k \rangle|^2 \right)^{-1} \approx Z_{n,c}^{2},\]
which grows like a positive power of the system size. Hence, the results of~\cite{MR3068390} show that the semicircle law is valid up to scales of order $M^{-1}$. This is consistent with the idea that the Dyson Brownian motion equilibrates globally for times of order $\oh(1)$. The analysis of~\cite{landonsosoeyauarxiv} can now be applied to prove universality of the $k$-point correlation function, i.e.\ the $k$-th marginal of the symmetrized eigenvalue density $\rho_{H_n}$:
\[\rho^{(k)}_{H_n}(\lambda_1, .. \lambda_k) =  \int_{\rr^{2^n - k}} \! \rho_{H_n} (\lambda_1, ..., \lambda_{2^n}) \, d\lambda_{k+1} ... \, d\lambda_{2^n}.\]
The precise theorem of~\cite{landonsosoeyauarxiv} is formulated in terms of
\begin{align*}\Psi_{n,E}^{(k)}(\alpha_1,...,\alpha_k) &= \rho^{(k)}_{H_n}\left( E + 2^{-n} \frac{\alpha_1}{\rho_{sc}(E)}, ..., E + 2^{-n} \frac{\alpha_k}{\rho_{sc}(E)} \right)\\
 &-  \rho^{(k)}_{GOE}\left( E + 2^{-n} \frac{\alpha_1}{\rho_{sc}(E)}, ..., E + 2^{-n} \frac{\alpha_k}{\rho_{sc}(E)} \right),
\end{align*}
where $\rho^{(k)}_{GOE}$ denotes the $k$-point correlation function of the $2^n \times 2^n$ GOE and $\rho_{sc}$ is the density of the semicircle law.
\begin{theorem}[cf.~\cite{MR3068390} and~\cite{landonsosoeyauarxiv}] Suppose $c < -1$, $E \in (-2,2)$ and $k \geq 1$. Then
\beq{rmtuniv}\lim_{n \to \infty} \int_{\rr^k} \! O(\alpha) \Psi_{n,E}^{(k)}(\alpha) \, d\alpha = 0\eeq
for every $O \in C_c^\infty(\rr^k)$.
\end{theorem} 
A corollary of the existence of a local law for the local resolvent on scale $M^{-1}$ is the complete delocalization of the eigenfunctions, in the sense that
\[\|\psi_{E,n}\|_\infty = \oh(M^{-1/2})\]
with overwhelming probability in mesoscopic spectral windows (see~\cite[Thm. 2.21]{0036-0279-66-3-R02}).

In case $c \in [-1,0)$, the density of states is no longer given by the semicircle law but one still expects the local statistics to exhibit random matrix universality. A technical difficulty in applying the results of~\cite{landonsosoeyauarxiv} is the lack of a lower bound on the non-averaged local density of states. This gap in the understanding of the local statistics in the ultrametric ensemble can be closed for the simpler Rosenzweig-Porter model, for which the techniques used in the proof of Theorem \ref{uepoissonthm} yield Poisson statistics if $t \ll N^{-1}$, and the results of~\cite{landonsosoeyauarxiv} can be applied to obtain random matrix universality for all $t \gg N^{-1}$. 
\begin{theorem}[\cite{landonsosoeyauarxiv} and~\cite{resflow}] Suppose $V_1, ... V_N$ are independent random variables with density $\rho \in L^\infty$ and $t = N^{-(1+c)}$ for some $c \in \rr$. Then, as $N \to \infty$,
\[H_N = \mbox{diag}(V_1, ..., V_N) + \Phi(t)\]
satisfies the following statements:
\begin{enumerate}
\item If $c > 0$, the level statistics $\mu_N$ near a Lebesgue point $E$ of $\rho$ converge in distribution to a Poisson point process with intensity $\rho(E)$.
\item If $c < 0$, the $k$-point correlation functions $\rho^{(k)}_N$ are universal in the sense of~\eqref{rmtuniv}.
\end{enumerate}
\end{theorem}
This theorem partially confirms the picture sketched in~\cite{0295-5075-115-4-47003,1367-2630-17-12-122002}, but leaves open the conjectured existence of non-ergodic delocalized states in the regime $N^{-1} \ll t \ll 1$. The existence of such a phase was recently proved in the independent works~\cite{benignibourgade} and~\cite{nonergodic}, complementing the eigenfunction bounds in~\cite{MR3134604} and~\cite{resflow}.

\providecommand{\bysame}{\leavevmode\hbox to3em{\hrulefill}\thinspace}
\providecommand{\MR}{\relax\ifhmode\unskip\space\fi MR }
\providecommand{\MRhref}[2]{%
  \href{http://www.ams.org/mathscinet-getitem?mr=#1}{#2}
}
\providecommand{\href}[2]{#2}

\end{document}